\documentclass{siamart190516}

\usepackage{graphicx}
\usepackage{latexsym,color}
\usepackage{amssymb}
\usepackage{amsfonts}
\usepackage{amsmath}

\newtheorem{thm}{Theorem}[section]

\newtheorem{rem}[thm]{Remark}

\newtheorem{exm}[thm]{Example}

\begin{document}

\title{Utility Indifference Pricing with High Risk Aversion and Small Linear Price Impact}
\author{Yan Dolinsky\footnotemark[1] and Shir Moshe\footnotemark[2]}
\date{\today}
\markboth{Y.Dolinsky and S.Moshe}{High Risk Aversion and Small Linear Price Impact}
\maketitle
\renewcommand{\thefootnote}{\fnsymbol{footnote}}
\footnotetext[1]
{Department of Statistics, Hebrew University of Jerusalem.
             \email{yan.dolinsky@mail.huji.ac.il}.}
\footnotetext[2]{Department of Statistics, Hebrew University of Jerusalem.
\email{shir.tapiro@mail.huji.ac.il}.}
\footnotetext[3] {Both authors supported by the GIF Grant 1489-304.6/2019 and the ISF grant 230/21.}

\renewcommand{\thefootnote}{\arabic{footnote}}

\pagenumbering{arabic}

\begin{abstract}
We consider the Bachelier model with
linear price impact.
Exponential utility indifference prices are studied for vanilla European options and
we compute their non-trivial scaling limit for a
vanishing price impact which is inversely proportional
to the risk aversion. Moreover, we find explicitly
a family of portfolios which are asymptotically optimal.
\end{abstract}

\section{Introduction}
In financial markets, trading moves prices against the trader:
buying faster increases execution prices, and selling faster decreases them.
This aspect of liquidity, known as market depth
\cite{B:86} or price-impact has recently received increasing attention
(see, for instance,
\cite{AFS:2010,BCE:2021,BLZ:2016,BT:2019,BV:2019,CHM:20,FSU:2019,GR:15,GW:2020,MMS:17,N:20,SZ:2018} and the references therein).

In \cite{GR:15} the authors showed that for a reasonable market model,
in the presence of price impact, super--replication
is prohibitively costly. Namely, in the presence of
price impact, even in market models such as the
Bachelier model or the Black–Scholes model
(which are complete in the frictionless setup)
there is no practical way to construct a hedging strategy
which eliminates all risk from a financial position.
This brings us to utility indifference pricing.

In this paper we study
utility indifference pricing for vanilla European options in
the multi--dimensional Bachelier model with
linear price impact.
Our main result is
computing the asymptotic behavior of the exponential
utility indifference prices
where the risk
aversion goes to infinity at a rate which is inversely proportional to
the linear price impact which goes to zero.
In addition we provide a family of
asymptotically optimal hedging strategies.

This type of scaling limits goes back to the seminal work of
Barles and Soner \cite{BS:98} which
determines the scaling limit of utility indifference prices
of vanilla options for small proportional
transaction costs and high risk aversion.
The present
note provides an analogous analysis for the case of linear price impact which results in quadratic transaction costs, albeit
using probabilistic techniques rather than taking a PDE
approach as pursued in \cite{BS:98}.

We divide the proof of our main result, namely Theorem \ref{thm.1} into two main
steps: the proof of the lower bound and the proof of the upper bound.

The proof of the lower bound goes through a dual representation of the certainty equivalent.
In the dual problem, there is only one player: a maximizing (adverse) player
that controls the probability measure.
A key ingredient in the proof of the lower bound is a construction
of a family of probability measures which attain (in the asymptotic sense) the desired limit. This is done in Proposition \ref{prop1}.

The proof of the upper bound does not use duality and is based on a direct argument (Proposition \ref{prop2}).
More precisely,
we construct a family
of trading strategies for which the expected utility converges to the scaling limit.
Roughly speaking, these strategies
are given by a reversion towards the $\Delta$--hedging
strategy which corresponds to a modified European option and a modified stock price.

The rest of the paper is organized as follows. In the next section we introduce
the setup and formulate the main results. In Section \ref{sec:3} we discuss the dual representation and prove the lower bound.
In Section \ref{sec:4} we prove the upper bound.

\section{Preliminaries and Main Results}\label{sec:2}
Let $T<\infty$ be the time horizon and let
$W=\left(W^1_t,...,W^d_t\right)_{t \in [0,T]}$
be a standard $d$-dimensional Brownian motion defined on the
filtered probability space $(\Omega, \mathcal{F},(\mathcal F_t)_{t\in [0,T]},\mathbb P)$
where $(\mathcal F_t)_{t\in [0,T]}$ is the
(augmented) filtration generated by $W$.
We consider a simple financial market
with a riskless savings account bearing zero interest (for simplicity) and with $d$-risky
asset $S=\left(S^1_t,...,S^d_t\right)_{t \in [0,T]}$ with Bachelier price dynamics
\begin{equation}\label{2.bac}
S^i_t=s^i_0+\mu^i t+\sum_{j=1}^d \sigma^{ij} W^j_t
\end{equation}
where $s_0=(s^1_0,...,s^d_0) \in \mathbb R^d$ is the initial position of the risky assets,
$\mu=(\mu^1,...\mu^d)\in\mathbb R^d$ is a constant vector (drift) and
$\sigma=\{\sigma^{ij}\}_{1\leq i,j \leq d}\in M_d(\mathbb R)$
is a constant nonsingular matrix (volatility). Without loss of generality we assume
that the constant nonsingular matrix $\sigma$ is a positive definite matrix.

Following
\cite{AlmgrenChriss:01}, we model the investor’s
market impact, in a temporary linear
form and, thus, when at time $t$ the investor turns over her position at the $i$--asset $\Phi^i_t$ at
the rate $\phi^i_t:=\dot{\Phi}^i_t$ the execution price is $S_t+\frac{\Lambda}{2}
\phi^i_t$ for some constant $\Lambda>0$.
As a result, the profits and
losses from a trading strategy $\phi=\left(\phi^1_t,...,\phi^d_t\right)_{t\in [0,T]}$ with the initial position
$\Phi_0=\left(\Phi^1_0,...,\Phi^d_0\right)$
are given by
\begin{equation}\label{eq:pnl}
V^{\Phi_0,\phi}_t:=\int_{0}^t \langle \Phi_u , dS_u\rangle-\frac{\Lambda}{2}\int_{0}^t ||\phi_v||^2 dv, \qquad t \in [0,T]
\end{equation}
where, for convenience, we assume that the investor marks to market her position
$\Phi_t=\Phi_0+\int_{0}^t \phi_v dv$
in the risky asset that she has acquired by time $t$.
As usual,
$\langle \cdot,\cdot\rangle$ and $||\cdot||$,
denotes the
standard scalar product and the Euclidean norm, respectively.
In our setup, the natural class of admissible strategies
is
$$
\mathcal A:=\left\{\phi=\left(\phi^1_t,...,\phi^d_t\right)_{t\in [0,T]}: \ \phi \text{ is } \
  \mathcal F\text{-adapted with } \int_{0}^T ||\phi_t||^2 dt<\infty
  \ \text{ a.s.}\right \}.
$$
\begin{rem}
Let us notice that by scaling the risky assets, there is no loss of generality in assuming that the constant
$\Lambda>0$ which represents the linear price impact is the same for all risky assets.
\end{rem}

Next, consider a vanilla European option with the payoff $f\left(S_T\right)$ where
$f:\mathbb R^d\rightarrow \mathbb R$ is a Lipschitz continuous function.
The investor will assess the quality of a hedge by the resulting expected utility.
 Assuming exponential utility with constant absolute risk aversion $\alpha>0$, the utility
indifference price and the certainty equivalent price of one unit of the claim $f\left(S_T\right)$ (see, e.g., \cite{R:08} for details
on indifference prices) do not depend on the investor's initial wealth
and, respectively, take the well-known forms
\begin{equation}\label{2.2}
\pi(\Lambda,\alpha,\Phi_0,f):=
\frac{1}{\alpha}\log\left(\frac{\inf_{\phi\in\mathcal A}\mathbb E_{\mathbb P}\left[\exp\left(\alpha\left(f(S_T)-V^{\Phi_0,\phi}_T\right)\right)\right]}
{\inf_{\phi\in\mathcal A}\mathbb E_{\mathbb P}\left[\exp\left(-\alpha V^{\Phi_0,\phi}_T\right)\right]}\right)
\end{equation}
and
$$c(\Lambda,\alpha,\Phi_0,f):=
\frac{1}{\alpha}\log\left(\inf_{\phi\in\mathcal A}\mathbb E_{\mathbb P}\left[\exp\left(\alpha\left(f(S_T)-V^{\Phi_0,\phi}_T\right)\right)\right]\right).$$

If the risk aversion $\alpha>0$ is fixed, then by applying standard density arguments we obtain that
for $\Lambda\downarrow 0$, the above indifference price converges to the unique price of the continuous
time complete (frictionless) market given by (\ref{2.bac}). A more interesting limit emerges,
however, if we re-scale the investor’s risk-aversion in the form $\alpha:=A/\Lambda$. Before we
formulate the limit theorem we need some preparations.

For a given $A>0$ introduce the functions
\begin{equation}\label{def1}
g^A(x):=\sup_{y\in\mathbb R^d}\left[f(x+y)-\frac{\langle y \sigma^{-1} ,y \rangle }{2\sqrt A }\right],\quad x=(x^1,...,x^d)\in\mathbb R^d
\end{equation}
and
\begin{equation}\label{def2}
u^A(t,x):=\mathbb E_{\mathbb P}\left[g^A(x+W_{T-t}\sigma )\right], \quad (t,x)\in [0,T]\times\mathbb R^d
\end{equation}
where the vectors $x,y,W$ are considered as row vectors and
for any row vector $z\in\mathbb R^d$, $z\sigma^{-1},z\sigma\in\mathbb R^d$ are the standard matrix products.
The term $u^A(t,S_t)$
represents the price at time $t$ of a European option with the payoff $g^A(S_T)$
in the complete market given by (\ref{2.bac}).
It is well known that $u\in C^{1,2}([0,T)\times\mathbb R^d)$ solves the PDE
\begin{equation}\label{PDE}
\frac{\partial u^A}{\partial t}+\frac{tr\left(\sigma^2 D^2_x u^A\right)}{2}=0 \ \ \ \ \mbox{in} \ \ [0,T)\times\mathbb R
\end{equation}
where
$tr (\cdot)$ is the trace of the square matrix $\cdot$ and
$D^2_x u^A$ is the Hessian matrix with respect to $x=(x^1,...,x^d)$ which is given by
$[D^2_x u^A]_{ij}:=\frac{\partial ^2 u^A}{\partial x^i \partial x^j}$, $1\leq i,j\leq d$.

For a given $A,\Lambda>0$ consider the $d$--dimensional (random) ODE ($\Phi$ is a row vector)
\begin{equation}\label{ODE}
\phi_t:=\dot{\Phi}_t=\frac{\sqrt A}{\Lambda}  \left(D_x u^A\left(t,S_t-\sqrt A \Phi_t\sigma\right)-\Phi_t\right)\sigma, \ \ \ \  t\in [0,T)
\end{equation}
where
$D_x u^A:=\left(\frac{\partial u^A}{\partial x^1},...,\frac{\partial u^A}{\partial x^d}\right)\in \mathbb R^d$ is the gradient with respect to $x$.
From the linear growth of $f$ it follows that for any $\epsilon>0$ the function
$D_x u^A,D^2_x u^A$ are uniformly bounded in the domain $[0,T-\epsilon]\times\mathbb R^d$. In particular
$D_x u^A$ is Lipschitz continuous with respect to $x$ in the domain $[0,T-\epsilon]\times\mathbb R^d$.
Hence, from the standard theory of
ODE
(see
Chapter II, Section 6 in \cite{W:98}) we obtain that
for a given initial value $\Phi_0$
there exists a unique solution to (\ref{ODE})
which we denote by $(\Phi^{A,\Lambda}_t)_{0\leq t<T}$.
Next, the Lipschitz continuity of $f$ implies that $g^A$ is a Lipschitz continuous function (with the same constant as $f$),
and so
$D_x u^A$ is uniformly
bounded in $[0,T)\times\mathbb R^d$. This together with the
mean reverting structure of the ODE (\ref{ODE}) yields that
$\lim_{t\rightarrow T-}\Phi^{A,\Lambda}_t$ exists and finite a.s.
Thus, we can
extend $\Phi^{A,\Lambda}$ to the interval $[0,T]$ by
$\Phi^{A,\Lambda}_t:=\lim_{t\rightarrow T-}\Phi^{A,\Lambda}_t$ and we
define $\phi^{A,\Lambda}\in\mathcal A$ by
$\phi^{A,\Lambda}_t: =\dot{\Phi}^{A,\Lambda}_t$
for $t<T$ and $\phi^{A,\Lambda}_T=0$.
Obviously,
$$\Phi^{A,\Lambda}_t:=\Phi_0+\int_{0}^t\phi^{A,\Lambda}_v dv, \ \ t\in [0,T]. $$
\begin{rem}
In words, the ODE (\ref{ODE}) says that the solution $\Phi^{A,\Lambda}$ is tracking
the $\Delta$--hedging strategy which corresponds to the modified payoff $g^A$ and the
shifted stock price $S_t-\sqrt A \Phi^{A,\Lambda}_t\sigma$. We notice that the shift depends on the solution
 $\Phi^{A,\Lambda}$.
\end{rem}

We arrive at the main result of the paper which
provides an explicit
formula for the asymptotic behavior of the certainty equivalent and
an optimal family (it should not be unique) of hedging strategies in the asymptotic sense.
 \begin{theorem}\label{thm.1}
For vanishing linear price impact $\Lambda \downarrow 0$ and re-scaled high risk-aversion
$A/\Lambda$ with $A>0$ fixed, the certainty equivalent of $f(S_T)$ has the scaling limit
\begin{equation}\label{2.3}
\lim_{\Lambda\downarrow 0} c(\Lambda,A/\Lambda,\Phi_0,f)=
u^A\left(0,s_0- \sqrt A\Phi_0\sigma\right)+\frac{\sqrt A\langle \Phi_0\sigma,\Phi_0\rangle }{2}.
\end{equation}
  Moreover, we have,
\begin{eqnarray*}
 &\lim_{\Lambda\downarrow 0}\frac{\Lambda}{A}
 \log\left(\mathbb E_{\mathbb P}\left[\exp\left(\frac{A}{\Lambda}\left(f(S_T)-V^{\Phi_0,\phi^{A,\Lambda}}_T\right)\right)\right]
\right)
\\
&=u^A\left(0,s_0- \sqrt A\Phi_0\sigma\right)+\frac{\sqrt A\langle \Phi_0\sigma,\Phi_0\rangle }{2}.
\end{eqnarray*}
\end{theorem}

From Theorem \ref{thm.1} we obtain immediately the following corollary which says that the
asymptotic value of the utility indifference prices is equal to the price of the vanilla European option with the payoff $g^A(S_T)$ and the shifted initial stock price
$s_0-\sqrt A\Phi_0\sigma$.
\begin{corollary}\label{cor.1}
For vanishing linear price impact $\Lambda \downarrow 0$ and re-scaled high risk-aversion
$A/\Lambda$ with $A>0$ fixed, the utility indifference price of $f(S_T)$ has the scaling limit
$$
 \lim_{\Lambda\downarrow 0}\pi(\Lambda,A/\Lambda,\Phi_0,f)=u^A\left(0,s_0- \sqrt A\Phi_0\sigma\right).$$
\end{corollary}
\begin{proof}
Apply (\ref{2.3}) and take $f\equiv 0$ for the
denominator of (\ref{2.2}).
\end{proof}

The following remark provides
a possible application of Corollary \ref{cor.1}.
\begin{rem}
Let
$\hat \pi(\Lambda,\alpha,\Phi_0,q):=\frac{1}{q} \pi(\Lambda,\alpha,\Phi_0,q f)$
be the per-unit indifference price for selling $q$ units of the option.
From (\ref{eq:pnl}) (apply the bijection $\phi\rightarrow \frac{\phi}{\Lambda}$)
it follows that
$\pi\left(\Lambda,\frac{A}{\Lambda},\Phi_0,f\right)=\hat\pi\left(\Lambda^2,A,\frac{\Phi_0}{\Lambda},\frac{1}{\Lambda}\right)$.
Hence, Corollary \ref{cor.1} can be viewed as a limit theorem for per-unit utility indifference prices for the case of
vanishing linear price impact and large position sizes. An interesting question is whether the theory
which was developed in \cite{R:17} can be
applied for obtaining
the asymptotic behaviour of
optimal position sizes (for the exact definition see \cite{R:17}). We leave this question for future research.
\end{rem}

We end this section with the following example.
\begin{exm}
Consider a European option with the payoff
$$f(x)=\left(\left\langle a,x\right\rangle+b\right)^{+}, \ \ x\in\mathbb R^d$$
for some constant $a\in \mathbb R^d$ and $b\in\mathbb R$.
Then we have
$$
g^A(x):=\sup_{y\in\mathbb R^d}\left[\left(\left\langle a,x+y\right\rangle+b  \right)^{+}-\frac{\left\langle y \sigma^{-1} ,y \right\rangle }{2\sqrt A }\right],\quad x\in\mathbb R^d.
$$

Clearly, the quadratic pattern
$y\rightarrow \langle a,y \rangle  -\frac{\langle y \sigma^{-1} ,y \rangle }{2\sqrt A }$
attains its maximum at $y^{*}:=\sqrt{A}a\sigma$. This together with the obvious inequality $g^A\geq f\geq 0$
yields that
$$g^A(x)=\left(\left\langle a,x+y^{*}\right\rangle+b- \frac{\left\langle y^{*} \sigma^{-1} ,y^{*} \right\rangle}{2\sqrt A }\right)^{+}=
\left(\left\langle a,x\right\rangle+b+\frac{\sqrt A \left\langle a \sigma ,a \right\rangle}{2}\right)^{+}.$$
\end{exm}
\begin{figure}
\centering
\includegraphics[width=0.9\textwidth]{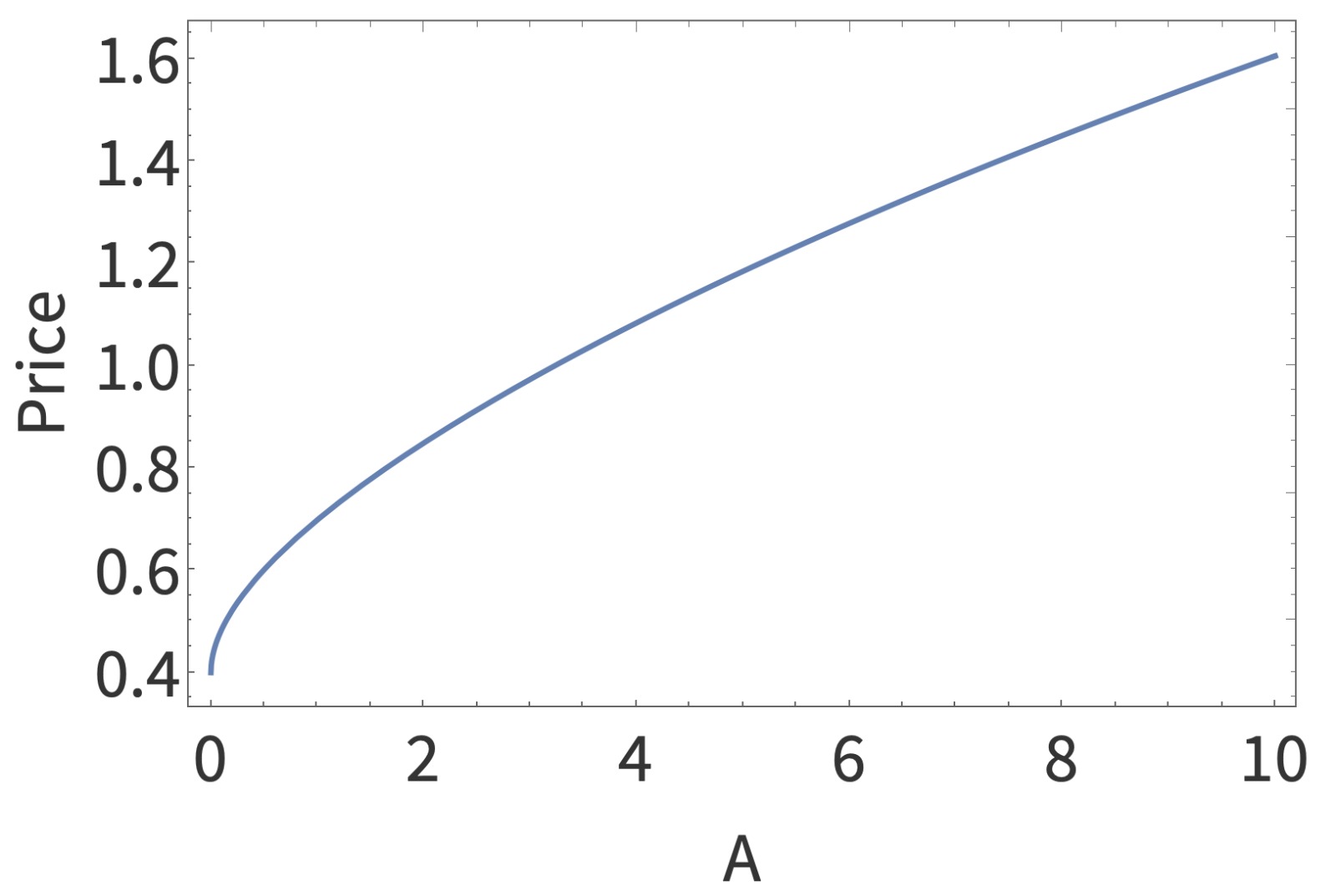}
\caption{We consider a one dimensional Bachelier model with $\sigma=1$. We assume that the initial number of stocks is $\Phi_0=0$.
The maturity date is $T=1$ and the initial stock price is $s_0=8$. Let $f(S_T)=(S_T-8)^{+}$, i.e. we take at the money call option.
We plot the limiting utility indifference price as a function of the parameter $A$. Namely we plot the function $A\rightarrow u^A(0,s_0)$.}
\end{figure}

\section{The Dual Problem and the Lower Bound}\label{sec:3}
In this section we establish the
inequality $\geq$ in (\ref{2.3}).

We start with the following lemma.
\begin{lemma}\label{lem1}
Denoting by $\mathcal Q$ the set of all probability measures $\mathbb Q\sim\mathbb P$ with finite
entropy
$\mathbb E_{\mathbb Q}\left[\log\left(\frac{d\mathbb Q}{d\mathbb P}\right) \right]<\infty$
relative to $\mathbb P$, we have
\begin{eqnarray}\label{3.2}
&c(\Lambda,\alpha,\Phi_0,f)\nonumber\\
&\geq\sup_{\mathbb Q\in\mathcal Q} \mathbb{E}_{\mathbb Q}\bigg[f\left(S_T\right)-\langle \Phi_0,S_T-s_0\rangle
-\frac{1}{\alpha}\log\left(\frac{d\mathbb Q}{d\mathbb P}\right)\\
&-\frac{1}{2\Lambda}\int_{0}^T\left|\left|S_t-
\mathbb{E}_{\mathbb Q}\left(S_T|\mathcal F_t\right)\right|\right|^2 dt\bigg].\nonumber
\end{eqnarray}
\end{lemma}
\begin{proof}
The proof rests on the classical Legendre-Fenchel duality inequality
\begin{equation}\label{dual}
 pq\leq e^p+q(\log q-1), \ \ \ p\in\mathbb R, \ q>0.
\end{equation}

Let $\phi\in\mathcal A$ and $\mathbb Q\in\mathcal Q$. From the Girsanov theorem it follows that there exists a process
$\theta=(\theta^1_t,...,\theta^d_t)_{t\in [0,T]}$ such that
$W^{\mathbb Q}_t:=W_t-\int_{0}^t\theta_v dv$, $t\in [0,T]$ is a $\mathbb Q$--Brownian motion. Moreover,
from the equality
\begin{equation*}\label{density}
Z:=\frac{d\mathbb Q}{d\mathbb P} = \exp\left(\int_0^T\langle \theta_t,
    dW_t\rangle -\frac{1}{2}\int_0^T||\theta_t||^2dt\right)
\end{equation*}
and the fact that $\mathbb Q\in\mathcal Q$ we obtain that
$$
\mathbb{E}_{\mathbb Q}\left[\log Z\right]=\frac{1}{2}\mathbb{E}_{\mathbb Q}\left[\int_{0}^T||\theta_t||^2 dt\right]<\infty
$$
and so from (\ref{2.bac})
\begin{equation}\label{fin}
\mathbb E_{\mathbb Q}\left[\sup_{0\leq t\leq T}||S_t||^2\right]<\infty.
\end{equation}
Without loss of generality we assume that
$\mathbb E_{\mathbb P}\left[e^{\alpha \left(f(S_T)-V^{\Phi_0,\phi}_T\right)}\right]<\infty.$
Thus, from (\ref{dual}) we obtain
\begin{equation}\label{fin1}
\alpha \mathbb E_{\mathbb Q}\left[\left(f(S_T)-V^{\Phi_0,\phi}_T\right)^{+}\right]\leq \mathbb E_{\mathbb P}\left[e^{\alpha \left(f(S_T)-V^{\Phi_0,\phi}_T\right)}\right]+
\mathbb{E}_{\mathbb Q}\left[\log Z\right]<\infty.
\end{equation}
Next, from (\ref{eq:pnl}) and the integration by parts formula it follows that
\begin{equation}\label{port}
V^{\Phi_0,\phi}_T=\langle \Phi_0,S_T-s_0\rangle+
\int_{0}^T\left(\langle \phi_t,S_T-S_t \rangle  -\frac{\Lambda}{2}\ ||\phi_t||^2 \right)dt.
\end{equation}
Hence, from the simple inequality
$$\langle \phi_t,S_T-S_t \rangle -\frac{\Lambda}{4}||\phi_t||^2\leq \frac{1}{\Lambda}||S_T-S_t||^2$$ we obtain
\begin{eqnarray*}
&f(S_T)-V^{\Phi_0,\phi}_T\geq -V^{\Phi_0,\phi}_T\\
&\geq -\langle \Phi_0,S_T-s_0\rangle+\int_{0}^T \left(\frac{\Lambda}{4}||\phi_t||^2-\frac{1}{\Lambda}||S_T-S_t||^2\right)dt.
\end{eqnarray*}
This together with (\ref{fin})--(\ref{fin1}) gives that
\begin{equation}\label{fin3}
\mathbb E_{\mathbb Q}\left[\int_{0}^T||\phi_t||^2dt \right]<\infty.
\end{equation}
From (\ref{dual})--(\ref{fin}) and (\ref{port})--(\ref{fin3})
we obtain
that for any $\gamma>0$
\begin{eqnarray}\label{3.gam}
&\mathbb E_{\mathbb P}\left[e^{\alpha \left(f(S_T)-V^{\Phi_0,\phi}_T\right)}\right]\nonumber\\
&\geq \mathbb E_{\mathbb P}\left[\alpha \gamma Z\left( f(S_T)-\langle \Phi_0,S_T-s_0\rangle -
\int_{0}^T\left(\langle \phi_t,S_T-S_t \rangle  -\frac{\Lambda}{2}\ ||\phi_t||^2 \right)dt\right)\right]\nonumber\\
&-\mathbb E_{\mathbb P}\left[\gamma Z\left(\log(\gamma Z)-1\right)
\right]\nonumber\\
&=\alpha \gamma\mathbb E_{\mathbb Q}\left[ f(S_T)-\langle \Phi_0,S_T-s_0\rangle\right]\nonumber\\
&-\alpha \gamma\mathbb E_{\mathbb Q}\left[\int_{0}^T\left(\langle \phi_t,\mathbb E_{\mathbb Q}[S_T|\mathcal F_t]-S_t \rangle -
\frac{\Lambda}{2}||\phi_t||^2 \right)dt\right]\nonumber\\
&-\gamma(\log \gamma-1)-\gamma \mathbb E_{\mathbb Q}\left[\log Z\right]\nonumber\\
&\geq \alpha \gamma\mathbb E_{\mathbb Q}\left[ f(S_T)-\langle \Phi_0,S_T-s_0\rangle-\frac{1}{2\Lambda}\int_{0}^T
\left|\left|\mathbb E_{\mathbb Q}[S_T|\mathcal F_t]-S_t\right|\right|^2 dt\right]\nonumber\\
&-\gamma(\log \gamma-1)-\gamma\mathbb E_{\mathbb Q}\left[\log Z\right]
\end{eqnarray}
where the last inequality follows from the maximization of the quadratic pattern
$$\phi_t\rightarrow \langle \phi_t,\mathbb E_{\mathbb Q}[S_T|\mathcal F_t]-S_t \rangle -
\frac{\Lambda}{2}||\phi_t||^2, \ \ \ \ t\in [0,T].$$

Optimizing (\ref{3.gam}) in $\gamma>0$ we arrive at
\begin{eqnarray*}
&\frac{1}{\alpha}\log\left(\mathbb E_{\mathbb P}\left[e^{\alpha \left(f(S_T)-V^{\Phi_0,\phi}_T\right)}\right]\right)\\
&\geq
\mathbb E_{\mathbb Q}\left[f(S_T)-\langle \Phi_0,S_T-s_0\rangle-\frac{1}{2\Lambda}\int_{0}^T\left|\left|\mathbb E_{\mathbb Q}[S_T|\mathcal F_t]-S_t\right|\right|^2 dt-\frac{1}{\alpha}\log Z\right].\nonumber
\end{eqnarray*}
Since $\phi\in\mathcal A$ and $\mathbb Q\in\mathcal Q$ were arbitrary we complete the proof.
\end{proof}

\begin{rem}
By mimicking the arguments of Proposition A.2 in \cite{BDR:21} to the multidimensional case
one can show that
the inequality in (\ref{3.2}) is in fact an equality. Namely, there is no duality gap. However,
since we need only the lower bound of the duality
then we just provide a self contained proof for (\ref{3.2}).
 \end{rem}

Next, we fix $A>0$ and prove the following key result.
\begin{proposition}\label{prop1}
Let $h:\mathbb R^d\rightarrow\mathbb R^d$
be a bounded and measurable function and let $Y=h(W_T)$.
Then,
\begin{equation*}
\lim\inf_{\Lambda\downarrow 0}c\left(\Lambda,A/\Lambda ,\Phi_0,f\right)\geq \mathbb E_{\mathbb P}\left[f\left(s_0+ W_T\sigma-Y\right)+\langle \Phi_0,Y\rangle-\frac{1}{2\sqrt A}\left\langle Y, Y\sigma^{-1}\right\rangle\right].
\end{equation*}
\end{proposition}
\begin{proof}
For a matrix valued process $\Xi=\{\Xi^{ij}_t\}_{1\leq i,j \leq d, 0\leq t\leq T}$ denote $\int_{0}^t \Xi_s\cdot dW_s$ the row vector
$(Z^1,...,Z^d)$ which is given by $Z^i=\int_{0}^t\sum_{j=1}^d \Xi^{ij}_s dW^j_s$.
By applying the martingale representation theorem and standard density arguments it follows that without loss of generality
we can assume that
$Y=y+\int_{0}^T H(s,W_s)\cdot dW_s$
where $y\in\mathbb R^d$ is a constant vector and
$H:[0,T]\times\mathbb R^d\rightarrow M_d(\mathbb R)$
is a continuous and bounded function
which satisfies
$H_{[T-\delta,T]\times\mathbb R^d}\equiv 0$
for some $\delta>0$.

The proof of Proposition \ref{prop1} is done in three steps. In the first step we fix $\Lambda>0$ and construct a sequence of probability measures
$\mathbb Q_n\sim\mathbb P$, $n\in\mathbb N$ with finite
entropy. In the second step we estimate (for the constructed probability measures) the asymptotic value of the components
in the right hand side of (\ref{3.2}) for $\alpha:=A/\Lambda$.
In the last step we summarize the computations and apply Lemma \ref{lem1}.\\

\textbf{Step I:}
Recall the drift vector $\mu=(\mu^1,...,\mu^d)$ which appears in (\ref{2.bac}).
For any $\Lambda>0$ and $0\leq s\leq t\leq T$ define
\begin{eqnarray*}
&K^{\Lambda}_{t,s}:=\cosh\left(\sqrt A (T-t)\sigma /\Lambda\right)
\left(\sinh\left(\sqrt A (T-s)\sigma/\Lambda\right)\right)^{-1}\\
&\mbox{and} \ \ a^{\Lambda}_t:=\mu \sigma^{-1}+
\frac{\sqrt A}{\Lambda} y K^{\Lambda}_{t,0}
\end{eqnarray*}
where for any square matrix $\xi$
$$\cosh(\xi):=\frac{\exp(\xi)+\exp(-\xi)}{2}, \ \ \sinh(\xi):=\frac{\exp(\xi)-\exp(-\xi)}{2}$$
and $\exp(\xi)$ is the matrix exponential of $\xi$.

Fix $\Lambda>0$. For any $n\in\mathbb N$  define the processes (which are also depend on $\Lambda$)
$W^n=(W^{n,1}_t,...,W^{n,d}_t)_{t\in [0,T]}$ and
$\theta^n=(\theta^{n,1}_t,...,\theta^{n,d}_t)_{t\in [0,T]}$
by the following recursive relations. For $t\in [0,T/n]$
$$
\theta^n_t:=a^{\Lambda}_t, \ \ W^n_t:=W_t+ \int_{0}^t \theta^n_v dv
$$
and for $k=1,...,n-1$,  $t\in (kT/n,(k+1)T/n]$
$$\theta^{n}_t:=a^{\Lambda}_t+\mathbb{I}_{||\kappa^n_t||<n}\kappa^n_t, \ \ \ W^n_t:=W_t+ \int_{0}^t \theta^n_v dv$$
where $\mathbb I$ denotes the indicator function and
$$\kappa^{n}_t:=\frac{\sqrt A}{\Lambda}
\int_{0}^{\frac{kT}{n}} K^{\Lambda}_{t,s}H(s,W^n_s)\cdot dW^n_s, \ \ \ \ t\in (kT/n,(k+1)T/n].
$$

Clearly, for any $n$ the process $\theta^n$
is a bounded process, and so
from the Girsanov theorem there exists a probability measure
$\mathbb Q^n$ such that $W^n$ is a $\mathbb Q^n$--Brownian motion.
For these probability measures we have the weak convergence (on the space of continuous function $C_d[0,T]$)
\begin{equation}\label{weakcon}
Q^n\circ \theta^n\Rightarrow \mathbb P\circ\left(a^{\Lambda}+\kappa^{\Lambda} \right)
\end{equation}
where $\kappa^{\Lambda}=(\kappa^{\Lambda}_t)_{t\in [0,T]}$ is given by
$\kappa^{\Lambda}_t:=\frac{\sqrt A}{\Lambda}
\int_{0}^{t} K^{\Lambda}_{t,s}H(s,W_s)\cdot dW_s.$
\\

\textbf{Step II:}
From the Fubini theorem, (\ref{2.bac}), (\ref{weakcon}) and
and the simple equality
$\frac{\sqrt A}{\Lambda}\int_{s}^T K^{\Lambda}_{t,s} dt =\sigma^{-1}$ $\forall s\in [0,T]$,
we get the weak convergence
\begin{equation}\label{weak1}
\mathbb Q^n\circ S_T\Rightarrow\mathbb P\circ \left(s_0+W_T\sigma -Y\right).
\end{equation}
Since $H$ is bounded
and $H_{[T-\delta,T]\times\mathbb R^d}\equiv 0$ then the term
$K_{t,s}H(s,W_s)$ is uniformly bounded and so we have the growth bound
\begin{equation}\label{boundd}
\sup_{n\in\mathbb N}\mathbb E_{\mathbb Q^n}\left[\sup_{0\leq t\leq T}||\kappa^n_t||^p\right]<\infty, \ \  \forall p>0.
\end{equation}
In particular,
$\mathbb Q^n\circ S_T$, $n\in\mathbb N$ are uniformly integrable. Thus, (\ref{weak1}) implies ($f$ has a linear growth)
\begin{equation}\label{weak2}
\lim_{n\rightarrow\infty}
\mathbb E_{\mathbb Q^n}\left[f\left(S_T\right)-\left\langle \Phi_0,S_T-s_0\right\rangle\right]=
\mathbb E_{\mathbb P}\left[f\left(s_0+W_T\sigma-Y\right)+\left\langle \Phi_0,Y\right\rangle \right].
\end{equation}

Next, let $H{'}$ be the transpose of $H$.
From the Fubini theorem,
 the It\^{o} isometry, (\ref{boundd}) and the equality $\mathbb E_{\mathbb Q^n}[\kappa^n_t]=0$ it  follows that
\begin{eqnarray}\label{2}
&\frac{\Lambda}{A}\lim_{n\rightarrow\infty}\mathbb E_{\mathbb Q^n}\left[\log\left(\frac{d\mathbb Q^n}{d\mathbb P}\right)\right]=
\frac{\Lambda}{2A}\lim_{n\rightarrow\infty}\mathbb E_{\mathbb Q^n}\left[\int_{0}^T ||\theta^n_t||^2dt\right]
\nonumber\\
&=\frac{\Lambda}{2A}\mathbb E_{\mathbb P}\left[\int_{0}^T||a^{\Lambda}_t+\kappa^{\Lambda}_t||^2 dt\right]
\leq c_1\Lambda+\frac{1}{2\Lambda}\int_{0}^T  ||y K^\Lambda_{t,0}||^2 dt\nonumber\\
&+\frac{1}{2\Lambda}\mathbb E_{\mathbb P}\left[tr\left(\int_{0}^T H'(s,W_s) \left(\int_{s}^T (K^{\Lambda}_{t,s})^2 dt\right) H(s,W_s)ds \right)\right]
\end{eqnarray}
for some constant $c_1>0$ which does not depend on $\Lambda$.

Finally, we estimate the last term in the right hand
side of (\ref{3.2}). From (\ref{2.bac})
\begin{eqnarray*}
&S_t-\mathbb E _{\mathbb Q^n}\left[S_T\left|\right.\mathcal F_t\right]\\
&=y G^{\Lambda}_t+\mathbb E _{\mathbb Q^n}\left[\int_{t}^T
\kappa^n_v \sigma dv\left|\right.\mathcal F_t\right]-\mathbb E _{\mathbb Q^n}\left[\int_{t}^T
\mathbb{I}_{||\kappa^n_v||\geq n}\kappa^n_v \sigma dv\left|\right.\mathcal F_t\right], \ \ t\in [0,T]
\end{eqnarray*}
for
$$G^{\Lambda}_t:=\frac{\sqrt A}{\Lambda}\sigma\int_{t}^T K^{\Lambda}_{v,0}dv=\sinh\left(\sqrt A (T-t)\sigma /\Lambda\right)\left(\sinh\left(\sqrt A T\sigma/\Lambda\right)\right)^{-1}.$$
By combining the
Doob inequality for the
martingales
$$\left(\mathbb E _{\mathbb Q^n}\left[\int_{0}^T
\mathbb{I}_{||\kappa^n_v||\geq n}||\kappa^n_v|| dv\left|\right.\mathcal F_t\right]\right)_{t\in [0,T]}, \ \ n\in\mathbb N$$
and (\ref{boundd})
it follows that $$\left(\mathbb E _{\mathbb Q^n}\left[\int_{t}^T
\mathbb{I}_{||\kappa^n_v||\geq n}\kappa^n_v \sigma dv\left|\right.\mathcal F_t\right]\right)_{t\in[0,T]}\rightarrow 0 \ \ \mbox{in} \ \ L^2(dt\otimes\mathbb Q^n).$$
This together with the equality $\mathbb E_{\mathbb Q^n}[\kappa^n_t]=0$ yields
\begin{eqnarray}
&\frac{1}{2\Lambda}\lim_{n\rightarrow\infty}\mathbb E_{\mathbb Q^n}\left[\int_{0}^T \left|\left|S_t-\mathbb E_{\mathbb Q^n}[S_T\left|\right.\mathcal F_t]\right|\right|^2 dt\right]\nonumber\\
&=\frac{1}{2\Lambda}\int_{0}^T \left|\left|yG^{\Lambda}_t\right|\right|^2 dt
+\frac{1}{2\Lambda}\lim_{n\rightarrow\infty}\mathbb E_{\mathbb Q^n}\left[\int_{0}^T \left|\left|\mathbb E_{\mathbb Q^n}\left[\int_{t}^T \kappa^n_v \sigma dv\left|\right.\mathcal F_t\right]\right|\right|^2 dt\right]\nonumber\\
&=\frac{1}{2\Lambda}\int_{0}^T \left|\left|yG^{\Lambda}_t\right|\right|^2 dt
+\frac{1}{2\Lambda}\mathbb E_{\mathbb P}\left[\int_{0}^T \left|\left|\mathbb E_{\mathbb P}\left[\int_{t}^T \kappa^{\Lambda}_v\sigma dv\left|\right.\mathcal F_t\right]\right|\right|^2 dt\right].
\end{eqnarray}
From the Fubini theorem
\begin{eqnarray*}
&\mathbb E_{\mathbb P}\left[\int_{t}^T \kappa^{\Lambda}_v dv\left|\mathcal F_t\right.\right]\nonumber\\
&=\frac{\sqrt A}{\Lambda}\mathbb E_{\mathbb P}\left[\int_{0}^T
\left(\int_{t\vee s}^T K^{\Lambda}_{v,s} dv\right)
H(s,W_s)\cdot dW_s
\left|\right.\mathcal F_t\right]\nonumber\\
&=\frac{\sqrt A}{\Lambda}
\int_{0}^t
\left(\int_{t}^T K^{\Lambda}_{v,s} dv\right)
H(s,W_s)\cdot dW_s.
\end{eqnarray*}
Hence, the It\^{o} isometry yields
\begin{eqnarray}\label{3}
&\frac{1}{2\Lambda}\mathbb E_{\mathbb P}\left[\int_{0}^T \left|\left|\mathbb E_{\mathbb P}\left[\int_{t}^T \kappa^{\Lambda}_v\sigma dv\left|\mathcal F_t\right.\right]\right|\right|^2 dt\right]\nonumber\\
&=\frac{1}{2\Lambda}\mathbb E_{\mathbb P}\left[tr \left(\int_{0}^T H'(s,W_s) \left(\int_{s}^T \left(L^{\Lambda}_{t,s}\right)^2dt\right) H(s,W_s)ds \right)\right]
\end{eqnarray}
where
$$L^{\Lambda}_{t,s}:=\frac{\sqrt A}{\Lambda}\sigma\int_{t}^T K^{\Lambda}_{v,s}dv=
\sinh\left(\sqrt A (T-t)\sigma /\Lambda\right)
\left(\sinh\left(\sqrt A (T-s)\sigma/\Lambda\right)\right)^{-1}.$$

\textbf{Step III:}
In this step we take $\Lambda\downarrow 0$. First, from the It\^{o} isometry
\begin{equation}\label{5}
\mathbb E_{\mathbb P}\left[\left\langle Y, Y\sigma^{-1}\right\rangle\right]= \langle y, y\sigma^{-1}\rangle+
\mathbb E_{\mathbb P}\left[tr\left(\int_{0}^T H'(s,W_s) \sigma^{-1} H(s,W_s)ds \right)\right].
\end{equation}
Next, since $\sigma$ is positive definite,
then for any $\epsilon>0$
we have the uniform convergence
\begin{eqnarray*}
&\lim_{\Lambda\downarrow 0}\left(\sup_{0\leq t\leq T-\epsilon}\left|\left|\frac{1}{2\Lambda}\int_{s}^T \left(K^{\Lambda}_{t,s}\right)^2dt-\frac{\sigma^{-1}}{4\sqrt A}\right|\right|\right)\nonumber\\
&=\lim_{\Lambda\downarrow 0}\left(\sup_{0\leq t\leq T-\epsilon}\left|\left|\frac{1}{2\Lambda}\int_{s}^T \left(L^{\Lambda}_{t,s}\right)^2dt-\frac{\sigma^{-1}}{4\sqrt A}\right|\right|\right)=0.
\end{eqnarray*}
Hence, by combining (\ref{3.2}) and (\ref{weak2})--(\ref{5}) (recall that
$H_{[T-\delta,T]\times\mathbb R^d}\equiv 0$) we obtain
\begin{eqnarray*}
&\lim\inf_{\Lambda\downarrow 0}c\left(\Lambda,A/\Lambda ,\Phi_0,f\right)\\
&\geq\lim\inf_{\Lambda\downarrow 0}\lim_{n\rightarrow\infty} \mathbb{E}_{\mathbb Q^n}\bigg[f\left(S_T\right)-\langle \Phi_0,S_T-s_0\rangle
-\frac{1}{\alpha}\log\left(\frac{d\mathbb Q^n}{d\mathbb P}\right)\\
&-\frac{1}{2\Lambda}\int_{0}^T\left|\left|S_t-
\mathbb{E}_{\mathbb Q^n}\left(S_T|\mathcal F_t\right)\right|\right|^2 dt\bigg]\\
&\geq\mathbb E_{\mathbb P}\left[f\left(s_0+W_T\sigma-Y\right)+\langle \Phi_0,Y\rangle-\frac{1}{2\sqrt A}\left\langle Y, Y\sigma^{-1}\right\rangle\right]
\end{eqnarray*}
as required.
\end{proof}
We now have all the pieces in place that we need for the
\textbf{completion of the proof of the inequality $"\geq"$ in (\ref{2.3})}.
\begin{proof}
Recall the definition of $g^A$ given
by $(\ref{def1})$.
Choose $\epsilon>0$.
From the Lipschitz continuity of $f$ it follows that there exists a finitely valued (and hence bounded and measurable)
function $\zeta:\mathbb R^d\rightarrow\mathbb R^d$ such that
\begin{equation}\label{final1}
g^A(x)<\epsilon+f\left(x-\zeta(x)\right)-\frac{\langle \zeta(x) \sigma^{-1} ,\zeta(x) \rangle }{2\sqrt A }, \ \ \ \forall x\in\mathbb R^d.
\end{equation}
By applying Proposition \ref{prop1} for the random variable
$$Y:=\zeta(s_0-\sqrt A\Phi_0+ W_T\sigma)+\sqrt A\Phi_0\sigma$$
and (\ref{final1}) for
$x:=s_0-\sqrt A\Phi_0\sigma+W_T\sigma$
we obtain
\begin{eqnarray*}
&\lim\inf_{\Lambda\downarrow 0} c(\Lambda,A/\Lambda,\Phi_0,f)\\
&\geq  \mathbb{E}_{\mathbb P}\bigg[g^A\left(s_0-\sqrt A\Phi_0\sigma+W_T\sigma\right)+
\frac{\left\langle Y\sigma^{-1}-\sqrt A\Phi_0,Y-\sqrt A\Phi_0\sigma \right\rangle}{2\sqrt A}\\
&+\langle\Phi_0,Y\rangle
-\frac{1}{2\sqrt A}\left\langle Y, Y\sigma^{-1}\right\rangle
\bigg]-\epsilon\\
&=u^A\left(0,s_0-\sqrt A\Phi_0\sigma\right)+\frac{\sqrt A\langle \Phi_0\sigma, \Phi_0\rangle }{2}-\epsilon
\end{eqnarray*}
where the equality follows from the definition of $u^A$ given by
(\ref{def2}).
By taking $\epsilon\downarrow 0$ we complete the proof.
\end{proof}

\section{Proof of the Upper Bound}\label{sec:4}
In order to complete the proof of Theorem \ref{thm.1} it remains to establish the following result.
\begin{proposition}\label{prop2}
Fix $A>0$ and recall the trading strategies $\Phi^{A,\Lambda}$, $\Lambda>0$ which are given by (\ref{ODE}).
Then,
\begin{eqnarray*}
&\lim_{\Lambda\downarrow 0}\frac{\Lambda}{A}
 \log\left(\mathbb E_{\mathbb P}\left[\exp\left(\frac{A}{\Lambda}\left(f(S_T)-V^{\Phi_0,\phi^{A,\Lambda}}_T\right)\right)\right]
\right)\\
&\leq u^A\left(0,s_0-\sqrt A \Phi_0\sigma\right)+\frac{\sqrt A\langle \Phi_0\sigma,\Phi_0\rangle }{2}.
\end{eqnarray*}
\end{proposition}
\begin{proof}
Introduce the $d$--dimensional process
$$\Theta^{A}_t:=D_x u^A\left(t,S_t-\sqrt A\Phi_t\sigma\right), \ \ t\in [0,T].$$
From the ODE (\ref{ODE}) it follows that for any $\Lambda>0$
$$\Phi^{A,\Lambda}_t=\Phi_0\cdot \exp\left(- \sqrt  A t\sigma/\Lambda\right)+
\frac{\sqrt A}{\Lambda}\int_{0}^t \Theta^{\Lambda}_v\sigma\exp\left(\sqrt A (v-t)\sigma /\Lambda\right) dv
, \ \  t\in[0,T].$$
Since $\sigma$ is positive definite and $D_x u^A$ is uniformly bonded ($f$ is Lipschitz), then there exists a constant $C>0$ such that
\begin{equation}\label{bound}
\sup_{0\leq t\leq T}||\Phi^{A,\Lambda}_t||\leq C, \ \ \forall \Lambda>0.
\end{equation}

Next,
for any $\Lambda>0$ define the process
$M^{\Lambda}=(M^{\Lambda}_t)_{t\in [0,T]}$ by
$$M^{\Lambda}_t:=\exp\left(\frac{A}{\Lambda}\left(u^A\left(t,S_t-\sqrt A \Phi^{A,\Lambda}_t\sigma\right)
+\frac{\sqrt A\langle \Phi^{A,\Lambda}_t\sigma,\Phi^{A,\Lambda}_t\rangle}{2}-
V^{\Phi_0,\phi^{A,\Lambda}}_t\right)\right).$$
From the It\^{o} formula, (\ref{eq:pnl}) and (\ref{PDE})
\begin{eqnarray*}
&\frac{dM^{\Lambda}_t}{M^{\Lambda}_t}=\frac{A}{\Lambda}\left\langle D_x u^A \left(t,S_t-\sqrt A \Phi^{A,\Lambda}_t\sigma\right)-
\Phi^{A,\Lambda}_t,dS_t\right\rangle\\
&+\frac{A^2}{2\Lambda^2}\left|\left|\left(D_x u^A
\left(t,S_t-\sqrt A \Phi^{A,\Lambda}_t\sigma\right)-\Phi^{A,\Lambda}_t\right)\sigma\right|\right|^2 dt\\
&-\frac{A^{3/2}}{\Lambda}\left\langle \phi^{A,\Lambda}_t,\left(D_x u^A
\left(t,S_t-\sqrt A\Phi^{A,\Lambda}_t\sigma\right)-\Phi^{A,\Lambda}_t\right)\sigma-\frac{\Lambda}{2\sqrt A}\phi^{A,\Lambda}_t\right\rangle
dt\\
&=\sqrt A\left\langle\phi^{A,\Lambda}_t\sigma^{-1} ,dS_t\right\rangle
\end{eqnarray*}
where the last equality follows from (\ref{ODE}).
Hence, from (\ref{2.bac})
$$\exp\left(-\sqrt A\left\langle \Phi^{A,\Lambda}_t-\Phi^{A,\Lambda}_0,\mu\sigma^{-1} \right \rangle
\right)M^{\Lambda}_t, \ \ t\in [0,T]$$
is a local--martingale, and so from the obvious inequality $M^{\Lambda}> 0$
we conclude that this process is a super--martingale.

Finally, (\ref{def1}) yields that
$
f(x)\leq g^A\left(x-y\sigma \right)+\frac{\langle y\sigma ,y \rangle }{2\sqrt A }$ for all
$x,y\in\mathbb R^d$.
This together with the growth bound
(\ref{bound}) and the above super--martingale property gives
(observe that $u^A(T,\cdot)=g^A(\cdot)$) that for any $\Lambda>0$
\begin{eqnarray*}
&\frac{\Lambda}{A}\log\left(\mathbb E_{\mathbb P}\left[\exp\left(\frac{A}{\Lambda}\left(f(S_T)-V^{\Phi_0,\phi^{\Lambda}}_T\right)\right)\right]\right)\\
&\leq \frac{\Lambda}{A}\log\left(\mathbb E_{\mathbb P}[M^{\Lambda}_T]\right)\leq
\frac{\Lambda}{A}\log \left(M^{\Lambda}_0\right)+\frac{\Lambda}{\sqrt A}2 C T||\mu\sigma^{-1}||\\
&=u^A\left(0,s_0-\sqrt A\Phi_0\sigma\right)+\frac{\sqrt A\langle \Phi_0\sigma,\Phi_0\rangle }{2}+\frac{\Lambda}{\sqrt A}2 C T||\mu\sigma^{-1}||
\end{eqnarray*}
and the result follows by taking $\Lambda\downarrow 0$.
\end{proof}

\section*{Acknowledgements}
The authors thank the AE and the anonymous reviewers for their valuable
reports and comments which helped to improve the quality of this paper.


\begin{thebibliography}{99}
\bibliographystyle{APT}
\footnotesize


\bibitem{AlmgrenChriss:01}
R. Almgren and N. Chriss,
{\em Optimal execution of portfolio transactions,}
Journal of Risk, {\bf 3,} 5--39, (2001).

\bibitem{AFS:2010}
A. Alfonsi, A. Fruth  and A. Schied,
{\em Optimal execution strategies in limit order books with general shape functions,}
Quantitative Finance, {\bf 2,} 143--157, (2010).

\bibitem{R:17}
M. Anthropelos, S. Robertson and K. Spliopoulus,
{\em The pricing of contingent claims and optimal positions in asymptotically complete markets,}
Annals of Applied Probability, {\bf 27,} 1178--1830, (2017).


\bibitem{B:86}
F. Black, {\em Noise,} Journal of Finance, {\bf 41,} 529--543, (1986).


\bibitem{BCE:2021}
E. Bayraktar, T. Caye, and I. Ekren,
{\em  Asymptotics for Small Nonlinear Price Impact: a PDE Approach to the Multidimensional Case,}
Mathematical Finance, {\bf 31,} 36--108, (2021).


\bibitem{BDR:21}
P. Bank, Y. Dolinsky and M. R\'{a}sonyi,
{\em What if we knew what the future brings?}
https://arxiv.org/abs/2108.04291, (2021).


\bibitem{BLZ:2016}
B. Bouchard, G. Loeper and Y.Zou,
{\em Almost-sure hedging with permanent price impact,}
Finance and Stochastics, {\bf 20,} 741--771, (2016).

\bibitem{BS:98}
G. Barles and H.M. Soner,
{\em Option pricing with transaction costs
and a nonlinear Black--Scholes equation,}
Finance and Stochastics, {\bf 2,} 369--397, (1998).


\bibitem{BT:2019}
B. Bouchard and X. Tan,
{\em Understanding the dual formulation for the hedging of
path-dependent options with price impact},
to appear in Annals of Applied Probability.
https://arxiv.org/pdf/1912.03946, (2020).

\bibitem{BV:2019}
P. Bank and M. Voß,
{\em Optimal Investment with Transient Price Impact,}
SIAM Journal on Financial Mathematics, {\bf 10,} 723--768, (2019).

\bibitem{R:08}
R. Carmona,
{\em Indifference pricing: theory and applications,}
Princeton University Press,
series in Financial Engineering, (2009).

\bibitem{CHM:20}
T. Caye, M. Herdegen and J. Muhle-Karbe,
{\em Trading with Small Nonlinear Price Impact,}
Annals of Applied Probability,
{\bf 30,} 706--746, (2020).

\bibitem{FSU:2019}
A. Fruth, T. Schöneborn and M. Urusov,
{\em Optimal trade execution in order books with stochastic liquidity,}
Mathematical Finance, {\bf 29,} 507--541, (2019).

\bibitem{GR:15}
P. Guasoni and M. R\'asonyi,
{\em Hedging, arbitrage and optimality under superlinear friction,}
Annals of Applied Probability, {\bf 25,} 2066--2095, (2015).

\bibitem{GW:2020}
P. Guasoni and M. Weber,
{\em Nonlinear price impact and portfolio choice,}
Mathematical Finance, {\bf 30,} 341--376, (2020).

\bibitem{MMS:17}
L. Moreau, J. Muhle-Karbe and H.M. Soner,
{\em Trading with Small Price Impact,}
Mathematical Finance, {\bf 27,} 350-400, (2017).

\bibitem{N:20}
S. Nadtochiy,
{\em A simple microstructural explanation of the concavity of price impact,}
to appear in Mathematical Finance. https://arxiv.org/pdf/2001.01860, (2020).

\bibitem{SZ:2018}
A. Schied and T. Zhang,
{\em A Market Impact Game Under Transient Price Impact,}
Mathematics of Operations Research, {\bf 44,} 102--121, (2019).


\bibitem{W:98}
W. Walter,
Ordinary Differential Equations,
Springer-Verlag New-York, (1998).


\end{thebibliography}
\end{document}